\documentclass[12pt, final]{l4dc2023}

\title[Interval Reachability of Neural Network Controllers]{Interval Reachability of Nonlinear Dynamical Systems with Neural Network Controllers}
\usepackage{times}
\usepackage{multicol} 
\usepackage{mathtools}
\usepackage{version,xspace}
\usepackage{url,doi}

 \newcommand{\setdef}[2]{\{#1
	\; | \; #2\}}

\usepackage{algorithm}
\usepackage{algpseudocode}

\newcommand\oprocendsymbol{\hbox{$\triangle$}}
\newcommand\oprocend{\relax\ifmmode\else\unskip\hfill\fi\oprocendsymbol}

\usepackage{enumitem} \renewcommand{\theenumi}{(\roman{enumi})}
\renewcommand{\labelenumi}{\theenumi}

\DeclareSymbolFont{bbold}{U}{bbold}{m}{n}
\DeclareSymbolFontAlphabet{\mathbbold}{bbold}

\newcommand{\real}{\mathbb{R}}

\newcommand{\seminorm}[1]{{\left\vert\kern-0.25ex\left\vert\kern-0.25ex\left\vert #1
		\right\vert\kern-0.25ex\right\vert\kern-0.25ex\right\vert}}

\newcommand{\semimeasure}[1]{\mu_{\seminorm{\cdot}}\kern-0.5ex\left(#1\right)}

\newcommand{\suchthat}{\;\ifnum\currentgrouptype=16 \middle\fi|\;}
\newcommand{\scirc}{\raise1pt\hbox{$\,\scriptstyle\circ\,$}}

\newcommand{\OF}{\mathsf{F}}
\newcommand{\OG}{\mathsf{G}}
\newcommand{\OH}{\mathsf{H}}

\newcommand{\ON}{\mathsf{N}}

\author{%
 \Name{Saber Jafarpour}\thanks{These authors contributed equally} \Email{saber@gatech.edu}\\
 \addr Georgia Institute of Technology
 \AND
 \Name{Akash Harapanahalli}\footnotemark[1] \Email{aharapan@gatech.edu}\\
 \addr Georgia Institute of Technology
 \AND
 \Name{Samuel Coogan} \Email{sam.coogan@gatech.edu}\\
 \addr Georgia Institute of Technology%
}

\begin{document}

\maketitle

\begin{abstract}%
This paper proposes a computationally efficient framework, based on interval analysis, for rigorous verification of nonlinear continuous-time dynamical systems with neural network controllers. 
Given a neural network, we use an existing verification algorithm to construct inclusion functions for its input-output behavior. 
Inspired by mixed monotone theory, we embed the closed-loop dynamics into a larger system using an inclusion function of the neural network and a decomposition function of the open-loop system.
This embedding provides a scalable approach for safety analysis of the neural control loop while preserving the nonlinear structure of the system. 

We show that one can efficiently compute hyper-rectangular over-approximations of the reachable sets using a single trajectory of the embedding system.  
We design an algorithm to leverage this computational advantage through partitioning strategies, improving our reachable set estimates while balancing its runtime with tunable parameters. 
%
We demonstrate the performance of this algorithm through two case studies. First, we demonstrate this method's strength in complex nonlinear environments. Then, we show that our approach  matches the performance of the state-of-the art verification algorithm for linear discretized systems.  

\end{abstract}


\section{Introduction}

Neural network components are increasingly deployed as controllers in  safety-critical applications such as self-driving vehicles and robotic systems. For example, they may be  designed based on reinforcement learning algorithms~\citep{TZ-GK-SL-PA:16} or trained to approximate some dynamic optimization-based controllers~\citep{SC-KS-NA-DDL-VK-GJP-MM:18}. 
However, neural networks are known to be vulnerable to small input perturbations~\citep{CZ-WZ-IS-JB-DE-IG-RF:13}; a slight disturbance in their input can lead to a large change in their output. 
In many applications, these learning-based components are interconnected with nonlinear and time-varying systems. Moreover, they are usually trained with no safety and robustness guarantees. Their learned control policies can therefore suffer from a significant degradation in performance when presented with uncertainties not accounted for during training. 
%

%



\paragraph*{Related works.}

There is extensive literature on verification of isolated neural networks. Rigorous verification approaches generally fall into three different categories: (i) reachability-based methods which focus on layer-by-layer estimation of reachable sets using interval bound propagation~\citep{MM-TG-MV:18,SG-etal:19,SW-KP-JW-JY-SJ:18}, activation function relaxation~\citep{HZ-etal:18}, and symbolic interval analysis~\citep{SW-KP-JW-JY-SJ:18}; (ii) optimization-based methods which use linear programming~\citep{EW-ZK:18}, semi-definite programming~\citep{MF-AR-HH-MM-GP:19a}, or search and optimization ~\citep{GK-CB-DLD-KJ-MJK:17} to estimate the input-output behavior of the neural networks; and (iii) probabilistic methods~\citep{JC-ER-ZK:19a,BL-CC-WW-LC:19}. We refer to the survey~\citep{CL-TA-CL-CS-CB-MJK:21} for a review of neural network verification algorithms.
Reachability of nonlinear dynamical systems has been studied using optimization-based methods such as the Hamilton-Jacobi approach~\citep{SB-MC-SH-CJT:17} and the level set approach~\citep{IM-CJT:00}. Recently, several computationally tractable approaches including the ellipsoidal method~\citep{ABK-PV:00}, the zonotope method~\citep{AG:05}, the mixed monotone reachability approach~\citep{SC-MA:15b} have been developed for reachability analysis.


It is well-known that, for a closed-loop system, a direct combination of state-of-the art neural network verification algorithms with the existing reachability analysis toolboxes can lead to overly-conservative estimates of reachable sets~\cite[Section 2.1]{SD-XC-SS:19}. For linear systems with neural network controllers, reachable set over-approximation has been studied using semi-definite programming~\citep{HH-MF-MM-GJP:20} and linear programming~\citep{ME-GH-CS-JPH:21}. Branch-and-Bound~\citep{TE-SS-MF:22} and suitable partitioning~\citep{ME-GH-JPH:21,WX-HDT-XY-TTJ:21} approaches have also been used to improve the accuracy of reachable set estimations.  For nonlinear systems with neural network controllers, ~\cite{CS-AM-AI-MJK:22} establishes a mixed integer programming framework for reachable set over-approximation using polynomial bounds of the system dynamics. ~\cite{SD-XC-SS:19} uses polynomial bounding on the dynamics as well as the neural network to over-approximate reachable sets. Other rigorous approaches for verification of closed-loop neural network controllers include using finite-state abstractions~\citep{XS-HK-YS:19} and structured zonotopes integrated into Taylor models~\citep{CS-MF-SG:22}.

\paragraph*{Contributions.}
In this paper, we use elements of mixed monotone system theory to develop a flexible framework for safety verification of nonlinear continuous-time systems with neural network controllers. 
First, we employ CROWN---a well-established isolated neural network verification framework~\citep{HZ-etal:18}---to construct an inclusion function for a pre-trained neural network's output. 
Then, we use this inclusion function and a decomposition function of the open-loop system to construct suitable embedding systems for the closed-loop system using twice the number of states.
Finally, we simulate trajectories of these embedding systems to compute hyper-rectangular over-approximations of the closed-loop reachable sets.
Our framework has several advantageous features. It is agnostic to the neural network verifier, only requiring an inclusion function for the neural network. Additionally, our approach is fast and scalable, as only a single simulation of the embedding system is required.
This feature, combined with a clever partitioning of the state space, is used to improve the accuracy of our reachable set over-approximations while retaining computational viability. Through several numerical experiments, we study the efficiency of our approach for obstacle avoidance in a simple vehicle model and for perturbation analysis in a linear quadrotor model.  

\section{Notation and Mathematical Preliminaries}

For every $x\in \real^n$ and every $r\in \real_{\ge 0}$ we define the closed ball $B_{\infty}(x,r) = \{y\in \real^n\mid \|y-x\|_{\infty}\le r\}$. The partial order $\le$ on $\real^n$ is defined by $x\le y$ if and only if $x_i\le y_i$, for every $i\in \{1,\ldots,n\}$. For every $x\le y$, we can define the interval $[x,y]=\{z\in \real^n\mid x\le z\le y\}$. The partial order $\le$ on $\real^n$ induces the southeast partial order $\le_{\mathrm{SE}}$ on $\real^{2n}$ defined by $\left[\begin{smallmatrix}x\\\widehat{x}\end{smallmatrix}\right]\le_{\mathrm{SE}} \left[\begin{smallmatrix}y\\\widehat{y}\end{smallmatrix}\right]$ if $x\le y$ and $\widehat{y}\le \widehat{x}$. We define the subsets 
$\mathcal{T}^{2n}_{\ge 0} = \setdef{\left[\begin{smallmatrix}
    x\\
    \widehat{x}
    \end{smallmatrix}\right]\in \real^{2n}}{x\le \widehat{x}}$ and $
  \mathcal{T}^{2n}_{\le 0} = \setdef{\left[\begin{smallmatrix}
    x\\
    \widehat{x}
    \end{smallmatrix}\right]\in \real^{2n}}{x\ge \widehat{x}}$ and $
    \mathcal{T}^{2n} = \mathcal{T}^{2n}_{\ge 0}\cup \mathcal{T}^{2n}_{\le 0}$.
For every two vectors $v,w\in \real^n$ and every $i\in \{1,\ldots,n\}$, we define the vector $v_{[i:w]}\in \real^n$ by $ \left(v_{[i:w]}\right)_j = \begin{cases}
    v_j & j\ne i\\
    w_j & j = i.
    \end{cases}$.
    Given a matrix
$B \in \mathbb{R}^{n\times m}$, we denote the non-negative part of $B$
by $[B]^+ = \max(B, 0)$ and the nonpositive part of $B$ by
$[B]^- = \min(B, 0)$. The Metzler and non-Metzler part of square matrix $A\in \real^{n\times n}$
are denoted by $\lceil A \rceil^{\mathrm{Mzl}}$ and
$\lfloor A \rfloor^{\mathrm{Mzl}}$, respectively, where $
  (\lceil A \rceil^{\mathrm{Mzl}})_{ij} =\begin{cases}
    A_{ij} & A_{ij} \geq 0\; \mbox{or} \; i =  j\\
    0 & \mbox{otherwise,}
  \end{cases}$ and $\lfloor A\rfloor^{\mathrm{Mzl}}= A-\lceil A
        \rceil^{\mathrm{Mzl}}$. Consider a control system 
        \begin{align}\label{eq:control}
        \dot{x} = f(x,u)
        \end{align}
         with state $x\in \real^n$, measurable disturbance $u:\real_{\ge 0}\to \real^m$. We denote the trajectory of the system with the disturbance $u$ starting from $x_0\in \real^n$ at time $t_0$ by $t\mapsto \phi_f(t,t_0,x_0,u)$. 
         Given an initial state set $\mathcal{X}_0\subseteq \real^n$ and a disturbance set $\mathcal{W}\subseteq\real^m$, we denote the reachable set of~\eqref{eq:control} by 
         \begin{align*}
\mathcal{R}_f(t,0,\mathcal{X}_0,\mathcal{U}) = \setdef{\phi_{f}(t,0,x_0,u)}{x_0\in \mathcal{X}_0,\;\;  u:\real_{\ge 0}\to \mathcal{U} \;\;\mbox{is piecewise cont.}}
 \end{align*}
The system~\eqref{eq:control} is continuous-time monotone on $\real^n$, if, for every $i\in\{1,\ldots,n\}$, every $x\le y\in \real^n$ with $x_i=y_i$, and every $u\le v$, we have $f_i(x,u) \le f_i(y,v)$.  
One can show that if a control system~\eqref{eq:control} is continuous-time monotone on $\real^n$, then its trajectories preserve the standard partial order by time, i.e., for every two trajectories $t\mapsto x_u(t)$ and $t\mapsto y_v(t)$ of the systems~\eqref{eq:control} with the control inputs $u$ and $v$, respectively, if $x(0)\le y(0)$ and $u\le v$, then $x_u(t)\le y_v(t)$ for all $t\in \real_{\ge 0}$~\citep{HLSS:95}.

 In general, computing the exact reachable sets of the system~\eqref{eq:control} is computationally intractable. Mixed monotone theory~\citep{GAE-HLS-EDS:06,HLS:08} has recently emerged as a computationally efficient framework for over-approximating the reachable sets of nonlinear systems~\citep{SC-MA:15b}. The control system~\eqref{eq:control} is mixed monotone if there exists a locally Lipschitz \emph{decomposition function} $\OF:\mathcal{T}^{2n}\times\mathcal{T}^{2m}\to \real^{n}$, satisfying
\begin{enumerate}[nosep]
\item\label{p1:decom} $\OF_i(x,x,u,u)=f_i(x,u)$, for every $x\in \real^n$ and every $u\in \real^p$;
\item\label{p2:decom} $\OF_i(x,\widehat{x},u,\widehat{u})\le
      \OF_i(y,\widehat{y},u,\widehat{u})$, for every $x\le y$ such
      that $x_i=y_i$, and every $\widehat{y}\le \widehat{x}$;
    \item\label{p3:decom} $\OF_i(x,\widehat{x},u,\widehat{u})\le \OF_i(x,\widehat{x},v,\widehat{v})$, for every $u\le v$ and every $\widehat{v}\le \widehat{u}$.
\end{enumerate}     
Several different approaches have proposed in the literature for constructing the decomposition function for the control system~\eqref{eq:control}~\citep{PJM-AD-MA:19a,MA-MD-SC:21,MA-SC:2021}.
Using a decomposition function $\OF$, we construct an embedding system associated to~\eqref{eq:control}:
\begin{gather} \label{eq:olembsys}
    \frac{d}{dt}\begin{bmatrix}x \\ \widehat{x}\end{bmatrix} = 
    \begin{bmatrix} \OF(x,\widehat{x},{u},\widehat{u}) \\ \OF(\widehat{x},{x},\widehat{u},{u}) 
    \end{bmatrix}
\end{gather}
 The key result in mixed monotone reachability theory is as follows: if $t\mapsto \left[\begin{smallmatrix}\underline{x}(t)\\ \overline{x}(t)\end{smallmatrix}\right]$ is the trajectory of the embedding system~\eqref{eq:olembsys} with disturbance $\left[\begin{smallmatrix}u\\ \widehat{u} \end{smallmatrix}\right] = \left[\begin{smallmatrix}\underline{u}\\ \overline{u} \end{smallmatrix}\right]$ starting from $\left[\begin{smallmatrix}\underline{x}_0\\ \overline{x}_0\end{smallmatrix}\right]$, then $\mathcal{R}_f(t,0,[\underline{x}_0,\overline{x}_0],[\underline{u},\overline{u}])\subseteq [\underline{x}(t),\overline{x}(t)]$, for every $t\in \real_{\ge 0}$~\cite[Proposition 6]{PJM-AD-MA:19a}. Given a map $g: \real^n\to \real^m$, the function $\left[\begin{smallmatrix}\underline{\OG}\\ \overline{\OG}\end{smallmatrix}\right]:\mathcal{T}_{\ge 0}^{2n}\to \mathcal{T}_{\ge 0}^{2m}$ is \emph{an inclusion function for $g$} if, for every $x\le \widehat{x}$ and every $z\in [x,\widehat{x}]$, we have $ \underline{\OG}(x,\widehat{x})\le g(z)\le \overline{\OG}(x,\widehat{x})$. 
Note that our definition of inclusion function is closely-related to the notion of \emph{inclusion interval function} in the interval analysis~\cite[Section 2.4.1]{LJ-MK-OD-EW:01}. 


\section{Problem Statement}

 We consider a nonlinear plant of the form
\begin{align}\label{eq:plant}
  \dot{x}&= f(x,u,w)
\end{align}
where $x\in \real^n$ is the state of the system, $u\in \real^p$ is the
control input, and $w\in \mathcal{W}\subseteq \real^q$ is the disturbance.  We assume that $f:\real^n\times \real^p\times \real^q\to \real^n$ is a parameterized vector field and the state feedback is parameterized by a $k$-layer
feed-forward neural network controller $\ON:\real^{n}\to \real^p$ defined by:
  \begin{align}\label{eq:NN}
   \xi^{(i)}(x) &= \phi^{(i)}(W^{(i-1)} \xi^{(i-1)}(x)+ b^{(i-1)}), \quad i\in \{1,\ldots,k\}\nonumber\\
     x &= \xi^{(0)}\qquad  u = W^{(k)}\xi^{(k)}(x) +b^{(k)} := \ON(x),
    \end{align}
 where $n_i$ is the number of neurons in the $i$th layer, $W^{(i-1)}\in \real^{n_i\times n_{i-1}}$ is the weight matrix of the $i$th layer, $b^{(i-1)}\in \real^{n_i}$ is the bias vector of the $i$th layer, $\xi^{(i)}(y)\in \real^{n_i}$ is the $i$-th layer hidden variable, and $\phi^{(i)}:\real^{n_{i}}\to \real^{n_i}$ is $i$th layer diagonal activation function satisfying $0\le
 \frac{\phi^{(i)}_j(x)-\phi^{(i)}_j(y)}{x-y}\le 1$, for every $j\in \{1,\ldots,n_{i}\}$. One can show that a large class of
 activation function including, but not restricted to, ReLU, leaky ReLU, sigmoid, and tanh satisfies this condition. We assume that, the neural network~\eqref{eq:NN} is trained to approximate an offline controller with the auxiliary objective of achieving a goal set $\mathcal{G}\subseteq \real^n$ while remaining in a safe set $\mathcal{S}\subseteq \real^n$. Our aim is to verify the safety of the closed-loop system given by:
 \begin{align}\label{eq:closedloop}
     \dot{x} &= f(x,\ON(x),w) := f^{\mathrm{cl}}(x,w).
 \end{align}
 i.e., to ensure that the closed-loop system avoids the unsafe domain $\real^n/\mathcal{S}$. 
 Given an initial state set $\mathcal{X}_0\subseteq \real^n$, our goal is to check if $\mathcal{R}_{f^{\mathrm{cl}}}(t,0,\mathcal{X}_0,\mathcal{W})\subseteq \mathcal{S}$ holds for every $t\in \real_{\ge 0}$. Our approach is based on constructing a computationally efficient over-approximation $\overline{\mathcal{R}}_{f^{\mathrm{cl}}}(t,0,\mathcal{X}_0,\mathcal{W})$ of the reachable set of the closed-loop system. Then, avoiding the unsafe set $\real^n/\mathcal{S}$ is guaranteed when $\overline{\mathcal{R}}_{f^{\mathrm{cl}}}(t,0,\mathcal{X}_0,\mathcal{W})\subseteq \mathcal{S}$, for every $t\in\real_{\ge 0}$.
 
 
 \section{Input-output reachability of neural networks}
 
 In order to estimate the input-output behavior of the neural network controller, we use the verification algorithm called CROWN~\citep{HZ-etal:18}. We first use CROWN to obtain an inclusion function for the input-output map of the neural network. Consider the neural network~\eqref{eq:NN} and let the perturbed input vector $x$ be in the interval $[\underline{x},\overline{x}]$. For every $i\in \{1,\ldots,k\}$, we define the pre-activation input to $i$th layer as $z^{(i)} = W^{(i)} \xi^{(i-1)}(x) + b^{(i)}$. When $x$ is perturbed in the interval $[\underline{x},\overline{x}]$, we assume that $L^{(i)},U^{(i)}\in \real^{n_i}$ are such that $L^{(i)} \le z^{(i)} \le U^{(i)}$. These bounds can be obtained using a fast neural network verification algorithm such as the Interval Bound Propagation (IBP) method~\citep[Equations (6) and (7)]{SG-etal:19}. In this case, for the $j$th neuron in the $i$th layer, there exist $\alpha^{(i)}_{U,j},\beta^{(i)}_{U,j},\alpha^{(i)}_{L,j},\beta^{(i)}_{L,j}$ such that 
 \begin{align}\label{eq:crownbound}
    \alpha^{(i)}_{L,j}(z + \beta^{(i)}_{L,j})  \le \phi_i(z) \le \alpha^{(i)}_{U,j}(z + \beta^{(i)}_{U,j}),\quad \mbox{ for every } L_j^{(i)}\le z \le U^{i}_j.
 \end{align}
 For every $x\in [\underline{x},\overline{x}]$, the output of the neural network is bounded as below:
    \begin{align}\label{eq:crown}
     \underline{A}(\underline{x},\overline{x})x + \underline{b}(\underline{x},\overline{x}) \le \ON(x) \le \overline{A}(\underline{x},\overline{x})x + \overline{b}(\underline{x},\overline{x}), 
 \end{align}
where functions  $\underline{A},\overline{A}: \mathcal{T}^{2n}_{\ge 0}\to \real^{p\times n}$  and $\underline{b},\overline{b}:\mathcal{T}^{2n}_{\ge 0}\to \real^p$ are defined as follows:
\begin{align}\label{eq:ABcrown}
    \overline{A}(\underline{x},\overline{x}) &= \Lambda^{(0)},\qquad \overline{b}(\underline{x},\overline{x}) = \sum\nolimits_{i=1}^{k} \Lambda^{(i)}b^{(i)} + \mathrm{diag}(\Lambda^{(i)}\Delta^{(i)}),\nonumber\\
    \underline{A}(\underline{x},\overline{x}) &= \Omega^{(0)},\qquad \underline{b}(\underline{x},\overline{x}) = \sum\nolimits_{i=1}^{k} \Omega^{(i)}b^{(i)} + \mathrm{diag}(\Omega^{(i)}\Theta^{(i)})
\end{align}
where, for every $i\in \{1,\ldots,k\}$, $\Lambda^{(i)},\Delta^{(i)},\Omega^{(i)},\Theta^{(i)}\in \real^{n_k\times n_i}$ are as defined in~\cite[Theorem 3.2]{HZ-etal:18} for the input perturbation set $B_{\infty}(\tfrac{\underline{x}+\overline{x}}{2},\tfrac{\overline{x}-\underline{x}}{2})$.

  \begin{theorem}[Inclusion functions for Neural Networks]\label{thm:crown-rectangle}
 Consider the $k$-layer feed-forward neural network $u=\ON(x)$ given by~\eqref{eq:NN}. Then, 
\begin{enumerate}
\item\label{p1:crown} for every input perturbation interval $[\underline{x},\overline{x}]$, the mapping $\left[\begin{smallmatrix}\underline{\OG}_{[\underline{x},\overline{x}]}\\ \overline{\OG}_{[\underline{x},\overline{x}]}\end{smallmatrix}\right]:\mathcal{T}_{\ge 0}^{2n}\to \mathcal{T}_{\ge 0}^{2p}$ defined by 
 \begin{align}\label{eq:crown-inclusion}
 \underline{\OG}_{[\underline{x},\overline{x}]}(x,\widehat{x}) &= [\underline{A}(\underline{x},\overline{x})]^{+} x  + [\underline{A}(\underline{x},\overline{x})]^{-} \widehat{x} + \underline{b}(\underline{x},\overline{x}),\nonumber\\
  \overline{\OG}_{[\underline{x},\overline{x}]}(x,\widehat{x}) & = [\overline{A}(\underline{x},\overline{x})]^{+} \widehat{x}  + [\overline{A}(\underline{x},\overline{x})]^{-}x + \overline{b}(\underline{x},\overline{x})
 \end{align}
 is an inclusion function for the neural network $\ON$ on $[\underline{x},\overline{x}]$; 
 \item\label{p2:crown} the mapping $\left[\begin{smallmatrix}\underline{\OH}\\ \overline{\OH}\end{smallmatrix}\right]:\mathcal{T}_{\ge 0}^{2n}\to \mathcal{T}_{\ge 0}^{2p}$ defined by $\underline{\OH}(x,\widehat{x}) = \underline{\OG}_{[x,\widehat{x}]}(x,\widehat{x})$ and $\overline{\OH}(x,\widehat{x}) = \overline{\OG}_{[x,\widehat{x}]}(x,\widehat{x})$ is an inclusion function for the neural network $\ON$ on $\real^n$. 
 \end{enumerate}
 \end{theorem}

 \begin{proof}
Regarding part~\ref{p1:crown}, suppose that $\eta\le \widehat{\eta}\in [\underline{x},\overline{x}]$ and $z\in [\eta,\widehat{\eta}]$. Then,
  \begin{align*}
     \overline{\OG}_{[\underline{x},\overline{y}]}(\eta,\widehat{\eta}) & = [\overline{A}(\underline{x},\overline{x})]^+\widehat{\eta} + [\overline{A}(\underline{x},\overline{x})]^-\eta + \overline{b}(\overline{x},\overline{x}) \ge  [\overline{A}(\underline{x},\overline{x})]^+ z + [\overline{A}(\underline{x},\overline{x})]^-z + \overline{b}(\overline{x},\overline{x}) \\ & = [\overline{A}(\underline{x},\overline{x})]z + \overline{b}(\overline{x},\overline{x}) \ge \ON(z),
 \end{align*}
 where the first inequality holds because $z\in [\eta,\widehat{\eta}]$, the matrix $[\overline{A}(\underline{x},\overline{x})]^+$ is non-negative, and the matrix $[\overline{A}(\underline{x},\overline{x})]^+$ is non-positive. The second inequality holds by noting the fact that $[\overline{A}(\underline{x},\overline{x})]^+ +[\overline{A}(\underline{x},\overline{x})]^- = [\overline{A}(\underline{x},\overline{x})]$, for every $\underline{x}\le \overline{x}$. Finally, the last inequality holds by~\eqref{eq:crown}. Similarly, one can show that $\underline{\OG}_{[\underline{x},\overline{x}]}(\widehat{\eta},\eta)\le \ON(z)$. Regarding part~\ref{p2:crown}, suppose that $x\le \widehat{x}$, for every $z\in [x,\widehat{x}]$, one can use a similar argument as part~\ref{p1:crown} to show that
 \begin{align*}
     \overline{\OG}_{[x,\widehat{x}]}(x,\widehat{x}) & = [\overline{A}(x,\widehat{x})]^+\widehat{x} + [\overline{A}(x,\widehat{x})]^-x + \overline{b}(x,\widehat{x}) \ge  [\overline{A}(x,\widehat{x})]^+ z + [\overline{A}(x,\widehat{x})]^-z + \overline{b}(x,\widehat{x}) \\ & = [\overline{A}(x,\widehat{x})]z + \overline{b}(x,\widehat{x}) \ge \ON(z).
 \end{align*}
 One can similarly show that $\underline{\OG}_{[x,\widehat{x}]}(\widehat{x},x)\le \ON(z)$, for every $z\in [x,\widehat{x}]$. Moreover, suppose that $\underline{x}=\overline{x} = z$. In this case, both inequalities in~\eqref{eq:crownbound} are equality with $\alpha^{(i)}_{L}=\alpha^{(i)}_U$ and $\beta^{(i)}_{L}=\beta^{(i)}_U$. Thus, using the equation~\eqref{eq:ABcrown}, we get that $\underline{A}(z,z)=\overline{A}(z,z)$ and $\underline{b}(z,z)=\overline{b}(z,z)$. This implies that $\ON(z) = \underline{A}(z,z)z + \underline{b}(z,z) = \underline{\OG}_{[z,z]}(z,z)$. Thus, $\left[\begin{smallmatrix}\underline{\OH}\\ \overline{\OH}\end{smallmatrix}\right]$ is an inclusion function for $\ON$ on the $\real^n$. 
 \end{proof}
 
  \begin{remark} 
The following remarks are in order. 
 \begin{enumerate}
     \item For each layer $i\in \{1,\ldots,k\}$, the intermediate pre-activation bounds $U^{i},L^{i}$ can be computed using either the Interval Bound Propagation (IBP)~\citep{SG-etal:19} or using the CROWN itself~\citep{HZ-etal:18}. In the former case, the IBP can be considered as a special case of CROWN by choosing $\alpha^{(i)}_U=\alpha^{(i)}_L=0$ and $\beta^{(i)}_L = L^{(i)}$ and $\beta^{(i)}_U = U^{(i)}$, for every $i\in \{1,\ldots,k\}$. 
     \item The inclusion function $\left[\begin{smallmatrix}\underline{\OG}_{[\underline{x},\overline{x}]}\\ \overline{\OG}_{[\underline{x},\overline{x}]}\end{smallmatrix}\right]$ defined in Theorem~\ref{thm:crown-rectangle}\ref{p1:crown} is linear but it depends on the input perturbation interval $[\underline{x}, \overline{x}]$. On the other hand, the inclusion function $\left[\begin{smallmatrix}\underline{\OH}\\ \overline{\OH}\end{smallmatrix}\right]$ defined in Theorem~\ref{thm:crown-rectangle}\ref{p2:crown} is nonlinear but it is independent of the input perturbation set. For both inclusion functions, it can be shown that the smaller the input perturbation interval, the tighter the relaxation bounds in equation~\eqref{eq:crownbound} are and, thus, the tighter the over- and the under-approximations of the neural network in equation~\eqref{eq:crown} are. 
 \end{enumerate}
 \end{remark}

  \section{Interval reachability of closed-loop system}
 
 In this section, inspired by mixed monotone reachability, we present a system-level approach for over-approximating the reachable set of the closed-loop system~\eqref{eq:closedloop} with the neural network controller~\eqref{eq:NN}. The key idea is to design a suitable function $E= \left[\begin{smallmatrix}\underline{E}\\ \overline{E}\end{smallmatrix}\right]:\mathcal{T}_{\ge 0}^{2n}\times \mathcal{T}_{\ge 0}^{2q}\to \real^{n}$ and use it to construct the following embedding system associated to the closed-loop system~\eqref{eq:closedloop}:
 \begin{align}\label{eq:closedloop-embedding}
     \frac{d}{dt}\begin{bmatrix}
      x\\
      \widehat{x}
     \end{bmatrix} = \begin{bmatrix}
      \underline{E}(x,\widehat{x}, w,\widehat{w})\\
      \overline{E}(\widehat{x},x,\widehat{w},w)
     \end{bmatrix}
 \end{align}
 Then, one can use the embedding system~\eqref{eq:closedloop-embedding} to study the propagation of the state and disturbance bounds with time. Suppose that the initial condition set is given by $\mathcal{X}_0\subseteq [\underline{x}_0,\overline{x}_0]$ and the disturbance set is given by $\mathcal{W}\subseteq [\underline{w},\overline{w}]$. Let $\OF$ be a decomposition function for the open-loop system~\eqref{eq:plant} and $\left[\begin{smallmatrix}\underline{\OG}\\ \overline{\OG}\end{smallmatrix}\right]$ and $\left[\begin{smallmatrix}\underline{\OH}\\ \overline{\OH}\end{smallmatrix}\right]$ be the inclusion functions of the neural network~\eqref{eq:NN} established in Theorem~\ref{thm:crown-rectangle}. We introduce
 ``global'' function $E^{\mathrm{G}}$, ``hybrid'' function $E^{\mathrm{H}}$, and ``local'' function $E^{\mathrm{L}}$ as follows:
  \begin{align} \label{eq:EGEHEL}
  \begin{aligned}
\begin{bmatrix}
    \underline{E}_i^{\mathrm{G}}(x,\widehat{x},w,\widehat{w})\\
    \overline{E}_i^{\mathrm{G}}(x,\widehat{x},w,\widehat{w})
\end{bmatrix} &= \begin{bmatrix}\OF_i(x,\widehat{x},\underline{\OH}(x,\widehat{x}),\overline{\OH}(x,\widehat{x}), w,\widehat{w})\\
      \OF_i(\widehat{x},x,\overline{\OH}(x,\widehat{x}),\underline{\OH}(x,\widehat{x}),\widehat{w},w)
      \end{bmatrix}\\
      \begin{bmatrix}
    \underline{E}_i^{\mathrm{H}}(x,\widehat{x},w,\widehat{w})\\
    \overline{E}_i^{\mathrm{H}}(x,\widehat{x},w,\widehat{w})
\end{bmatrix} &= \begin{bmatrix} \OF_i(x,\widehat{x},\underline{\OG}_{[x,\widehat{x}]}(x,\widehat{x}_{[i:x]}),\overline{\OG}_{[x,\widehat{x}]}(x,\widehat{x}_{[i:x]}), w,\widehat{w})\\
\OF_i(\widehat{x},x,\overline{\OG}_{[x,\widehat{x}]}(x_{[i:\widehat{x}]},\widehat{x}),\underline{\OG}_{[x,\widehat{x}]}(x_{[i:\widehat{x}]},\widehat{x}), \widehat{w},w)
      \end{bmatrix}\\
      \begin{bmatrix}
    \underline{E}_i^{\mathrm{L}}(x,\widehat{x},w,\widehat{w})\\
    \overline{E}_i^{\mathrm{L}}(x,\widehat{x},w,\widehat{w})
\end{bmatrix} &= \begin{bmatrix} \OF_i(x,\widehat{x},\underline{\OH}(x,\widehat{x}_{[i:x]}),\overline{\OH}(x,\widehat{x}_{[i:x]}), w,\widehat{w})\\
      \OF_i(\widehat{x},x,\overline{\OH}(x_{[i:\widehat{x}]},\widehat{x}),\underline{\OH}(x_{[i:\widehat{x}]},\widehat{x}), \widehat{w},w)
      \end{bmatrix},
      \end{aligned}
 \end{align}
 for every $i\in\{1,\ldots,n\}$. 


 \begin{theorem}[Interval over-approximation of reachable sets]\label{thm:main}
 Consider the control system~\eqref{eq:plant} with the neural network controller~\eqref{eq:NN}. Suppose that the initial condition $x_0$ belongs to $\mathcal{X}_0\subseteq [\underline{x}_0,\overline{x}_0]$ and the disturbance $w$ belongs to $\mathcal{W}\subseteq [\underline{w},\overline{w}]$. Let $\OF:\mathcal{T}^{2n}\times \mathcal{T}^{2p}\times \mathcal{T}^{2q}\to \real^n$ be a decomposition function for the open-loop system~\eqref{eq:plant}.  Given $s\in \{\mathrm{G},\mathrm{L},\mathrm{H}\}$, consider  $E^{\mathrm{G}},E^{\mathrm{H}},E^{\mathrm{L}}$ as defined in~\eqref{eq:EGEHEL} and let $t\mapsto \left[\begin{smallmatrix}\underline{x}^{s}(t)\\ \overline{x}^{s}(t)\end{smallmatrix}\right]$ be the trajectory of the embedding system~\eqref{eq:closedloop-embedding} with $E = E^s$ and disturbance $\left[\begin{smallmatrix}w\\ \widehat{w} \end{smallmatrix}\right] = \left[\begin{smallmatrix}\underline{w}\\ \overline{w} \end{smallmatrix}\right]$ starting from $\left[\begin{smallmatrix}\underline{x}_0\\ \overline{x}_0\end{smallmatrix}\right]$.  Then, the following statements hold:  
 \begin{enumerate}
     \item\label{p2:reachestimate} for every $s\in \{\mathrm{G},\mathrm{L},\mathrm{H}\}$ and every $t\in \real_{\ge 0}$, we have 
     \begin{align*}
         \mathcal{R}_{f^{\mathrm{cl}}}(t,0,\mathcal{X}_0,\mathcal{W})\subseteq [\underline{x}^{s}(t),\overline{x}^{s}(t)],
     \end{align*}
\end{enumerate}

\begin{enumerate}\setcounter{enumi}{1}
    \item\label{p4:order} If all the activation functions of the neural network~\eqref{eq:NN} are ReLU, then, for every $t\in \real_{\ge 0}$, 
 \begin{align*}
     \mathcal{R}_{f^{\mathrm{cl}}}(t,0,\mathcal{X}_0,\mathcal{W})\subseteq [\underline{x}^{\mathrm{L}}(t),\overline{x}^{\mathrm{L}}(t)]\subseteq [\underline{x}^{\mathrm{H}}(t),\overline{x}^{\mathrm{H}}(t)]\subseteq  [\underline{x}^{\mathrm{G}}(t),\overline{x}^{\mathrm{G}}(t)]. 
 \end{align*}
\end{enumerate}
 \end{theorem}
 \begin{proof}
Regarding part~\ref{p2:reachestimate}, we provide the proof for $s=\mathrm{H}$. The proof of other cases are mutadis mutandis with $s=\mathrm{H}$. Note that, by~\citep[Theorem 1]{MA-MD-SC:21}, the tight decomposition function $d^{c}$ for the closed-loop system $f^{\mathrm{cl}}$ and the tight decomposition function $d^{o}$ for the open loop system $f$ can be computed as follows:
 \begin{align}\label{eq:tight-closedloop}
     d_i^{c}(x,\widehat{x},w,\widehat{w}) &= \begin{cases}
     \min_{z\in [x,\widehat{x}],\xi\in [w,\widehat{w}]\atop z_i=x_i} f_i(z,\ON(z),\xi), & x\le \widehat{x},w\le\widehat{w}\\
     \max_{z\in [\widehat{x},x],\xi\in [w,\widehat{w}]\atop z_i=\widehat{x}_i} f_i(z,\ON(z),\xi), &  \widehat{x}\le x , \widehat{w}\le w.
     \end{cases}\\
     d_i^{o}(x,\widehat{x},u,\widehat{u},w,\widehat{w}) &= \begin{cases}
     \min_{z\in [x,\widehat{x}],\xi\in [w,\widehat{w}]\atop z_i=x_i, \eta\in [u,\widehat{u}]} f_i(z,\eta,\xi), & x\le\widehat{x} , u\le\widehat{u}, w\le\widehat{w}, \\
     \max_{z\in [\widehat{x},x],\xi\in [w,\widehat{w}]\atop z_i=\widehat{x}_i,\eta\in [u,\widehat{u}]} f_i(z,\eta,\xi), & \widehat{x}\le x , \widehat{u}\le u, \widehat{w}\le w.
     \end{cases}
 \end{align}
 The trajectory of the embedding system~\eqref{eq:closedloop-embedding} with $d = d^{\mathrm{c}}$ and disturbance $\left[\begin{smallmatrix}w\\ \widehat{w} \end{smallmatrix}\right] = \left[\begin{smallmatrix}\underline{w}\\ \overline{w} \end{smallmatrix}\right]$ starting from $\left[\begin{smallmatrix}\underline{x}_0\\ \overline{x}_0\end{smallmatrix}\right]$  is denoted by $t\mapsto \left[\begin{smallmatrix}\underline{x}^{c}(t)\\ \overline{x}^{c}(t)\end{smallmatrix}\right]$. 
 Since $\OF$ is a decomposition function for the open-loop system $f$, for every $x\le \widehat{x}$, $u\le \widehat{u}$, $w\le \widehat{w}$, and $i\in \{1,\ldots,n\}$, we get
\begin{align}\label{eq:usefulinequality}
  \OF_i(x,\widehat{x},u,\widehat{u},w,\widehat{w}) \le  d_i^{o}(x,\widehat{x},u,\widehat{u},w,\widehat{w}) =   \min_{z\in [x,\widehat{x}], z_i=x_i\atop \xi\in [w,\widehat{w}], \eta\in [u,\widehat{u}]} f_i(z,\eta,\xi).
\end{align}
Suppose that $x\le \widehat{x}$ and let $i\in \{1,\ldots,n\}$ and $y \in [x,\widehat{x}]$ be such that $y_i=x_i$. By Theorem~\ref{thm:crown-rectangle}\ref{p1:crown},  $\left[\begin{smallmatrix} \underline{\OG}_{[x,\widehat{x}]} \\ \overline{\OG}_{[x,\widehat{x}]}\end{smallmatrix}\right]$ is an inclusion function for $\ON$ on $[x,\widehat{x}]$. Moreover, $y\in [x,\widehat{x}_{[i:x]}]\subseteq [x,\widehat{x}]$ and thus
 \begin{align}\label{eq:keyinequality}
     \underline{\OG}_{[x,\widehat{x}]}(x,\widehat{x}_{[i:x]}) \le \ON(y) \le \overline{\OG}_{[x,\widehat{x}]}(x,\widehat{x}_{[i:x]}).  
 \end{align}
Therefore, for every $i\in\{1,\ldots,n\}$, we get
 \begin{align}\label{eq:inequality-tight}
    d_i^{c}(x,\widehat{x},w,\widehat{w}) & = \min_{z\in [x,\widehat{x}],\xi\in [w,\widehat{w}]\atop z_i=x_i} f_i(z,\ON(z),\xi) \ge \min_{z\in [x,\widehat{x}],\xi\in [w,\widehat{w}], z_i=x_i\atop u\in [\underline{\OG}_{[x,\widehat{x}]}(x,\widehat{x}_{[i:x]}),\overline{\OG}_{[x,\widehat{x}]}(x,\widehat{x}_{[i:x]})]} f_i(z,u,\xi) \nonumber \\ & \ge \OF_i(x,\widehat{x}, \underline{\OG}_{[x,\widehat{x}]}(x,\widehat{x}_{[i:x]}), \overline{\OG}_{[x,\widehat{x}]}(x,\widehat{x}_{[i:x]}), w,\widehat{w}) = \underline{E}_i^{\mathrm{H}}(x,\widehat{x},w,\widehat{w}).
 \end{align}
where the first equality holds by definition of $d^{c}$, the second inequality holds by equation~\eqref{eq:keyinequality}, the third inequality holds by equation~\eqref{eq:usefulinequality}, and the fourth inequality holds by definition of $\underline{E}^{\mathrm{H}}$. Similarly, one can show that $d_i^{c}(\widehat{x},x,\widehat{w},w) \le \overline{E}_i^{\mathrm{H}}(\widehat{x},x,\widehat{w},w)$, for every $i\in \{1,\ldots,n\}$. This implies that $\left[\begin{smallmatrix}\underline{E}^{\mathrm{H}}(x,\widehat{x},w,\widehat{w})\\\overline{E}^{\mathrm{H}}(\widehat{x},x,\widehat{w},w)\end{smallmatrix}\right] \le_{\mathrm{SE}} \left[\begin{smallmatrix}d^{c}(x,\widehat{x},w,\widehat{w})\\ d^{c}(\widehat{x},x,\widehat{w},w)\end{smallmatrix}\right]$, for every $x\le \widehat{x}$ and every $w\le \widehat{w}$. Note that, by~\cite[Theorem 1]{MA-MD-SC:21}, the vector field $\left[\begin{smallmatrix}d^{c}(x,\widehat{x},w,\widehat{w})\\ d^{c}(\widehat{x},x,\widehat{w},w)\end{smallmatrix}\right]$ is monotone with respect to the southeast order $\le_{\mathrm{SE}}$ on $\real^{2n}$.  Now, we can use~\cite[Theorem 3.8.1]{ANM-LH-DL:08}, to deduce that $\left[\begin{smallmatrix}\underline{x}^{\mathrm{H}}(t)\\\overline{x}^\mathrm{H}(t)\end{smallmatrix}\right] \le_{\mathrm{SE}} \left[\begin{smallmatrix}\underline{x}^{c}(t)\\\overline{x}^c(t)\end{smallmatrix}\right]$, for every $t\in \real_{\ge 0}$. This implies that $[\underline{x}^{c}(t),\overline{x}^{c}(t)]\subseteq [\underline{x}^{\mathrm{H}}(t),\overline{x}^{\mathrm{H}}(t)]$. On the other hand, by~\citep[Theorem 2]{MA-MD-SC:21}, we know that $\mathcal{R}_{f^{cl}}(t,\mathcal{X}_0,\mathcal{W})\subseteq [\underline{x}^{c}(t),\overline{x}^{c}(t)]$, for every $t\in \real_{\ge 0}$. This lead to $\mathcal{R}_{f^{cl}}(t,\mathcal{X}_0,\mathcal{W})\subseteq [\underline{x}^{\mathrm{H}}(t),\overline{x}^{\mathrm{H}}(t)]$, for every $t\in \real_{\ge 0}$.

 Regarding part~\ref{p4:order}, suppose that all the activation functions are ReLU and let $L^{(i)}_{[x,\widehat{x}]}$ and $U^{(i)}_{[x,\widehat{x}]}$ are the intermediate bounds associated to the input perturbation $[x,\widehat{x}]$. Then, for every $\left[\begin{smallmatrix}x\\ \widehat{x}\end{smallmatrix}\right]\le_{\mathrm{SE}} \left[\begin{smallmatrix}y \\ \widehat{y}\end{smallmatrix}\right]\in \mathcal{T}^{2n}_{\ge 0}$, the intermediate bounds in CROWN satisfy:
 \begin{align*}
     L^{(i)}_{[x,\widehat{x}]} \le L^{(i)}_{[y,\widehat{y}]}, \qquad U^{(i)}_{[y,\widehat{y}]}\le U^{(i)}_{[x,\widehat{x}]}, \qquad\mbox{ for every } i\in \{1,\ldots,k\}.
 \end{align*}
 Since all the activation functions are ReLU, we can use~\cite[Table 1 and Table 2]{HZ-etal:18} to show that
\begin{align*}
     \alpha^{(i)}_{L,[x,\widehat{x}]} &\le  \alpha^{(i)}_{L,[y,\widehat{y}]}, \qquad \alpha^{(i)}_{U,[y,\widehat{y}]} \le  \alpha^{(i)}_{U,[x,\widehat{x}]}, \\
     \beta^{(i)}_{L,[x,\widehat{x}]} &\le  \beta^{(i)}_{L,[y,\widehat{y}]}, \qquad \beta^{(i)}_{U,[y,\widehat{y}]} \le  \beta^{(i)}_{U,[x,\widehat{x}]},  \qquad \mbox{ for every } i\in \{1,\ldots,k\}
 \end{align*}
 Therefore, using the equations~\eqref{eq:crownbound} and~\eqref{eq:ABcrown}, and the formula in~\cite[Theorem 3.2]{HZ-etal:18} for $\Lambda^{(i)},\Delta^{(i)},\Omega^{(i)},\Theta^{(i)}$, we get
 \begin{align*}
    \underline{A}(x,\widehat{x}) &\le  \underline{A}(y,\widehat{y}), \qquad \overline{A}(y,\widehat{y}) \le  \overline{A}(x,\widehat{x}), \\
     \underline{b}(x,\widehat{x}) & \le  \underline{b}(y,\widehat{y}), \qquad \overline{b}(y,\widehat{y}) \le \overline{b}(x,\widehat{x}),
 \end{align*}
 This implies that 
 \begin{align*}
     \underline{\OG}_{[x,\widehat{x}]}(w,\widehat{w}) &= [\underline{A}(x,\widehat{x})]^+ w + [\underline{A}(x,\widehat{x})]^- \widehat{w}  \le   [\underline{A}(y,\widehat{y})]^+ w + [\underline{A}(y,\widehat{y})]^- \widehat{w} \\ &\le [\underline{A}(y,\widehat{y})]^+ v + [\underline{A}(y,\widehat{y})]^- \widehat{v}  = \underline{\OG}_{[y,\widehat{y}]}(v,\widehat{v}).
 \end{align*}
 Similarly, we can show that $\overline{\OG}_{[y,\widehat{y}]}(v,\widehat{v}) \le  \overline{\OG}_{[x,\widehat{x}]}(w,\widehat{w})$. As a result, for every $\left[\begin{smallmatrix}x\\ \widehat{x}\end{smallmatrix}\right]\le_{\mathrm{SE}} \left[\begin{smallmatrix}y \\ \widehat{y}\end{smallmatrix}\right]\in \mathcal{T}^{2n}_{\ge 0}$ and $\left[\begin{smallmatrix}w\\\widehat{w}\end{smallmatrix}\right]\le_{\mathrm{SE}} \left[\begin{smallmatrix}v\\\widehat{v}\end{smallmatrix}\right]\in \mathcal{T}^{2n}_{\ge 0}$,
 \begin{align}\label{eq:inequalitymain}
 \left[\begin{smallmatrix}\underline{\OG}_{[x,\widehat{x}]}(w,\widehat{w})\\  \overline{\OG}_{[x,\widehat{x}]}(w,\widehat{w}) \end{smallmatrix}\right] \le_{\mathrm{SE}} \left[\begin{smallmatrix}\underline{\OG}_{[y,\widehat{y}]}(v,\widehat{v})\\  \overline{\OG}_{[y,\widehat{y}]}(v,\widehat{v}) \end{smallmatrix}\right] 
 \end{align}
  Note that, using property~\eqref{eq:inequalitymain}, one can easily check that $E^{\mathrm{G}},E^{\mathrm{L}},E^{\mathrm{H}}$ are continuous-time monotone vector fields with respect to the southeast order $\le_{\mathrm{SE}}$ on $\real^{2n}$. On the other hand, we have $\widehat{x}_{[i:x]} \le \widehat{x}$. Therefore, using~\eqref{eq:inequalitymain},
  \begin{align}\label{eq:inequalitymain2}
\left[\begin{smallmatrix}\underline{\OG}_{[x,\widehat{x}]}(x,\widehat{x})\\ 
 \overline{\OG}_{[x,\widehat{x}]}(x,\widehat{x}) \end{smallmatrix}\right]
 \le_{\mathrm{SE}} 
 \left[\begin{smallmatrix}\underline{\OG}_{[x,\widehat{x}]}(x,\widehat{x}_{[i:x]})\\  \overline{\OG}_{[x,\widehat{x}]}(x,\widehat{x}_{[i:x]}) \end{smallmatrix}\right],\qquad \mbox{for every } x\le \widehat{x}
 \end{align}
 This implies that, for every $x\le \widehat{x}$ and every $w\le \widehat{w}$, and every $i\in\{1,\ldots,n\}$, we have
 \begin{align*}
     \underline{E}_i^{\mathrm{G}}(x,\widehat{x},w,\widehat{w}) &=  \OF_i(x,\widehat{x},\underline{\OG}_{[x,\widehat{x}]}(x,\widehat{x}),\overline{\OG}_{[x,\widehat{x}]}(x,\widehat{x}), w,\widehat{w}) \\ & \le  \OF_i(x,\widehat{x},\underline{\OG}_{[x,\widehat{x}]}(x,\widehat{x}_{[i:x]}),\overline{\OG}_{[x,\widehat{x}]}(x,\widehat{x}_{[i:x]}), w,\widehat{w}) = \underline{E}_i^{\mathrm{H}}(x,\widehat{x},w,\widehat{w})
 \end{align*}
 where the inequality holds using equation~\eqref{eq:inequalitymain2} and the fact that $\OF$ is a decomposition function for the open-loop system~\eqref{eq:plant}. Similarly, we can show that $\overline{E}_i^{\mathrm{H}}(\widehat{x},x,\widehat{w},w)\le \overline{E}_i^{\mathrm{G}}(\widehat{x},x,\widehat{w},w)$, for every $x\le \widehat{x}$ and every $w\le \widehat{w}$ and every $i\in\{1,\ldots,n\}$. This implies that $\left[\begin{smallmatrix}\underline{E}^{\mathrm{G}}(x,\widehat{x},w,\widehat{w})\\ \overline{E}^{\mathrm{G}}(x,\widehat{x},w,\widehat{w})\end{smallmatrix}\right] \le_{\mathrm{SE}} \left[\begin{smallmatrix}\underline{E}^{\mathrm{H}}(x,\widehat{x},w,\widehat{w})\\ \overline{E}^{\mathrm{H}}(x,\widehat{x},w,\widehat{w})\end{smallmatrix}\right]$, for every $x\le \widehat{x}$ and every $w\le \widehat{w}$. Note again that the vector field $E^{\mathrm{G}}$ is continuous-time monotone with respect to the southeast order $\le_{\mathrm{SE}}$ on $\real^{2n}$. Now, we can use~\cite[Theorem 3.8.1]{ANM-LH-DL:08}, to deduce that $\left[\begin{smallmatrix}\underline{x}^{\mathrm{G}}(t)\\\overline{x}^\mathrm{G}(t)\end{smallmatrix}\right] \le_{\mathrm{SE}} \left[\begin{smallmatrix}\underline{x}^{\mathrm{H}}(t)\\\overline{x}^{\mathrm{H}}(t)\end{smallmatrix}\right]$, for every $t\in \real_{\ge 0}$. This implies that $[\underline{x}^{\mathrm{H}}(t),\overline{x}^{\mathrm{H}}(t)]\subseteq [\underline{x}^{\mathrm{G}}(t),\overline{x}^{\mathrm{G}}(t)]$.  The proofs for the other inclusions are similar and we remove it for the sake of brevity. 
 \end{proof}

When the open-loop system~\eqref{eq:plant} is linear with $f(x,u,w)=Ax + B u + C w$ where $A\in \real^{n\times n}$, $B\in \real^{n\times p}$, and $C\in \real^{n\times q}$, one can find another embedding function for the closed-loop system~\eqref{eq:closedloop}:
 \begin{align}\label{eq:Elin}
 E^{\mathrm{Lin}}(x,\widehat{x},w,\widehat{w}) = \left[\begin{smallmatrix}\lceil A + R(x,\widehat{x}) \rceil^{\mathrm{Mzl}} & \lfloor A + R(x,\widehat{x}) \rfloor^{\mathrm{Mzl}} \\ \lfloor A + S(x,\widehat{x}) \rfloor^{\mathrm{Mzl}} & \lceil A + S(x,\widehat{x}) \rceil^{\mathrm{Mzl}}\end{smallmatrix}\right]\left[\begin{smallmatrix}
 x\\ \widehat{x}\end{smallmatrix}\right] + \left[\begin{smallmatrix}
 C^+ & C^-\\C^-& C^+\end{smallmatrix}\right]\left[\begin{smallmatrix}
 w\\ \widehat{w}\end{smallmatrix}\right],
 \end{align}
 where $R(x,\widehat{x})= B^{+}[\underline{A}(x,\widehat{x})] + B^{-}[\overline{A}(x,\widehat{x})]$ and $S(x,\widehat{x})= B^{+}[\overline{A}(x,\widehat{x})] + B^{-}[\underline{A}(x,\widehat{x})]$. The next theorem compares reachable set over-approximation using $E^{\mathrm{Lin}}$ and hybrid function $E^{\mathrm{H}}$.
 
 \begin{theorem}[Linear systems]\label{eq:linear}
 Consider the control system~\eqref{eq:plant} with $f(x,u,w)=Ax + B u + C w$ with $A\in \real^{n\times n}$, $B\in \real^{n\times p}$, and $C\in \real^{n\times q}$ and with the neural network controller~\eqref{eq:NN}. Suppose that the initial set satisfies $\mathcal{X}_0\subseteq [\underline{x}_0,\overline{x}_0]$ and the disturbance set satisfies $\mathcal{W}\subseteq [\underline{w},\overline{w}]$. Then,
 \begin{enumerate}
     \item\label{p1} We recover all the results of Theorem~\eqref{thm:main} with the following ``hybrid'' function: 
 \begin{align*}  
      E^{\mathrm{H}}(x,\widehat{x},w,\widehat{w}) = \left[\begin{smallmatrix}\lceil A\rceil^{\mathrm{Mzl}} + \lceil R(x,\widehat{x}) \rceil^{\mathrm{Mzl}} & \lfloor A\rfloor^{\mathrm{Mzl}} + \lfloor R(x,\widehat{x}) \rfloor^{\mathrm{Mzl}} \\ \lfloor A \rfloor^{\mathrm{Mzl}}+ \lfloor S(x,\widehat{x}) \rfloor^{\mathrm{Mzl}} & \lceil A \rceil^{\mathrm{Mzl}} + \lceil S(x,\widehat{x}) \rceil^{\mathrm{Mzl}}\end{smallmatrix}\right]\left[\begin{smallmatrix}
 x\\ \widehat{x}\end{smallmatrix}\right] + \left[\begin{smallmatrix}
 C^+ & C^-\\C^-& C^+ \end{smallmatrix}\right]\left[\begin{smallmatrix}
 w\\ \widehat{w}\end{smallmatrix}\right].
 \end{align*}
 where $R(x,\widehat{x})= B^{+}[\underline{A}(x,\widehat{x})] + B^{-}[\overline{A}(x,\widehat{x})]$ and $S(x,\widehat{x})= B^{+}[\overline{A}(x,\widehat{x})] + B^{-}[\underline{A}(x,\widehat{x})]$.
     \item\label{p2} if $t\mapsto \left[\begin{smallmatrix}\underline{x}^{\mathrm{Lin}}(t)\\ \overline{x}^{\mathrm{Lin}}(t)\end{smallmatrix}\right]$ is the trajectory of the embedding system~\eqref{eq:Elin} with disturbance $\left[\begin{smallmatrix}w\\ \widehat{w} \end{smallmatrix}\right] = \left[\begin{smallmatrix}\underline{w}\\ \overline{w} \end{smallmatrix}\right]$ starting from $\left[\begin{smallmatrix}\underline{x}_0\\ \overline{x}_0\end{smallmatrix}\right]$ then, for every $t\in \real_{\ge 0}$,
 \begin{align*}
    \mathcal{R}_{f^{\mathrm{cl}}}(t,0,\mathcal{X}_0,\mathcal{W})\subseteq [\underline{x}^{\mathrm{Lin}}(t),\overline{x}^{\mathrm{Lin}}(t)]\subseteq [\underline{x}^{\mathrm{H}}(t),\overline{x}^{\mathrm{H}}(t)].
     \end{align*}
 \end{enumerate} 
 \end{theorem}
 \begin{proof}
 Note that a decomposition function for the open-loop system is given by $\OF(x,\widehat{x},u,\widehat{u},w,\widehat{w}) = \lceil A \rceil^{\mathrm{Mzl}} x + \lfloor A \rfloor^{\mathrm{Mzl}}\widehat{x} + B^{+}u + B^{-}\widehat{u} + C^{+}w + C^{-} \widehat{w}$~\cite[Example 3]{SC:20}. Using simple algebraic manipulations, one can show that for every $i\in \{1,\ldots,n\}$ and every $M,L\in \real^{n\times n}$,
 \begin{align}
\left(M^{+}x + M^{-}\widehat{x}_{[i:x]}\right)_i &=  \left(\lceil M \rceil^{\mathrm{Mzl}} x + \lfloor M \rfloor^{\mathrm{Mzl}}\widehat{x}\right)_i, \label{eq:manipulation-1}\\
\left(\lceil M \rceil^{\mathrm{Mzl}} + \lceil L \rceil^{\mathrm{Mzl}}\right) x & + \left(\lfloor M \rfloor^{\mathrm{Mzl}} + \lfloor L \rfloor^{\mathrm{Mzl}}\right)\widehat{x} \le \lceil M + L \rceil^{\mathrm{Mzl}} x + \lfloor M +L \rfloor^{\mathrm{Mzl}}\widehat{x}. \label{eq:manipulation-2} 
 \end{align}
Thus, using Theorem~\ref{thm:main} and the identity~\eqref{eq:manipulation-1}, for every $i\in \{1,\ldots,n\}$, we have
 \begin{align*}
     \underline{E}_i^{\mathrm{H}}(x,\widehat{x},w,\widehat{w}) & = \left(\lceil A \rceil^{\mathrm{Mzl}} x + \lfloor A \rfloor^{\mathrm{Mzl}}\widehat{x} + B^{+} [\underline{A}]^{+}(x,\widehat{x})x +  B^{+}[\underline{A}]^{-}(x,\widehat{x})\widehat{x}_{[i:x]} \right)_i
     \\ & \qquad  + \left( B^{-}[\overline{A}]^{+}(x,\widehat{x})x +  B^{-}[\overline{A}]^{-}(x,\widehat{x})\widehat{x}_{[i:x]} + C^{+}w + C^{-} \widehat{w}  \right)_i
     \\  & = \left(\lceil A \rceil^{\mathrm{Mzl}} x + \lfloor A \rfloor^{\mathrm{Mzl}}\widehat{x} + \lceil B^{+} \underline{A}(x,\widehat{x}) \rceil^{\mathrm{Mzl}}x +  \lfloor B^{+}\underline{A}(x,\widehat{x})\rfloor^{\mathrm{Mzl}}\widehat{x}\right)_i 
     \\ & \qquad  + \left(\lfloor B^{-}\overline{A}(x,\widehat{x})\rfloor^{\mathrm{Mzl}}x +  \lceil B^{-}\overline{A}(x,\widehat{x})\rceil^{\mathrm{Mzl}}\widehat{x} + C^{+}w + C^{-} \widehat{w}\right)_i
     \\ &= \left(\lceil A \rceil^{\mathrm{Mzl}} x + \lceil B^{+}\underline{A}(x,\widehat{x}) + B^{-}\overline{A}(x,\widehat{x})\rceil^{\mathrm{Mzl}} x\right)_i \\ & + \left(\lfloor A + \lfloor \widehat{x} + \lfloor B^{+}[\overline{A}(x,\widehat{x})]  + B^{-}[\underline{A}(x,\widehat{x})]\rfloor^{\mathrm{Mzl}}\widehat{x} \right)_i
     \\ & \qquad +  \left(C^{+}w + C^{-} \widehat{w}\right)_i \\ &  = \left(\lceil A \rceil^{\mathrm{Mzl}} x + \lceil R(x,\widehat{x}) \rceil^{\mathrm{Mzl}}x + \lfloor A \rfloor^{\mathrm{Mzl}} \widehat{x} + \lfloor S(x,\widehat{x}) \rfloor^{\mathrm{Mzl}} \widehat{x} + C^{+}w + C^{-} \widehat{w}\right)_i. 
 \end{align*}
 This completes the proof of part~\ref{p1}. Regarding part~\ref{p2}, using the identity~\eqref{eq:manipulation-2}, for every $i\in \{1,\ldots,n\}$, we have
 \begin{align*}
 \underline{E}_i^{\mathrm{H}}(x,\widehat{x},w,\widehat{w}) & = \left(\lceil A \rceil^{\mathrm{Mzl}} x + \lceil R(x,\widehat{x}) \rceil^{\mathrm{Mzl}}x + \lfloor A \rfloor^{\mathrm{Mzl}} \widehat{x} + \lfloor S(x,\widehat{x}) \rfloor^{\mathrm{Mzl}} \widehat{x} + C^{+}w + C^{-} \widehat{w}\right)_i \\ & \le \left(\lceil A + R(x,\widehat{x}) \rceil^{\mathrm{Mzl}}x + \lfloor A + S(x,\widehat{x}) \rfloor^{\mathrm{Mzl}} \widehat{x} + C^{+}w + C^{-} \widehat{w}\right)_i = \underline{E}_i^{\mathrm{Lin}}(x,\widehat{x},w,\widehat{w})
 \end{align*}
 Similarly, one can show that,  $\overline{E}_i^{\mathrm{Lin}}(x,\widehat{x},w,\widehat{w}) \le \overline{E}_i^{\mathrm{H}}(x,\widehat{x},w,\widehat{w})$, for every $x\le \widehat{x}$ and every $w\le \widehat{w}$ and every $i\in\{1,\ldots,n\}$. This implies that $\left[\begin{smallmatrix}\underline{E}^{\mathrm{H}}(x,\widehat{x},w,\widehat{w})\\ \overline{E}^{\mathrm{H}}(x,\widehat{x},w,\widehat{w})\end{smallmatrix}\right] \le_{\mathrm{SE}} \left[\begin{smallmatrix}\underline{E}^{\mathrm{Lin}}(x,\widehat{x},w,\widehat{w})\\ \overline{E}^{\mathrm{Lin}}(x,\widehat{x},w,\widehat{w})\end{smallmatrix}\right]$, for every $x\le \widehat{x}$ and every $w\le \widehat{w}$. By part~\ref{p1}, it is easy to show tat the vector field $E^{\mathrm{H}}$ is continuous-time monotone with respect to the southeast order $\le_{\mathrm{SE}}$ on $\real^{2n}$. Now, we can use~\cite[Theorem 3.8.1]{ANM-LH-DL:08}, to deduce that $\left[\begin{smallmatrix}\underline{x}^{\mathrm{H}}(t)\\\overline{x}^\mathrm{H}(t)\end{smallmatrix}\right] \le_{\mathrm{SE}} \left[\begin{smallmatrix}\underline{x}^{\mathrm{Lin}}(t)\\\overline{x}^{\mathrm{Lin}}(t)\end{smallmatrix}\right]$, for every $t\in \real_{\ge 0}$. This implies that $[\underline{x}^{\mathrm{Lin}}(t),\overline{x}^{\mathrm{Lin}}(t)]\subseteq [\underline{x}^{\mathrm{H}}(t),\overline{x}^{\mathrm{H}}(t)]$. The other inclusion follows from Theorem~\ref{thm:main}. 
 \end{proof}
 
 \begin{remark} The following remarks are in order.
 \begin{enumerate}
 
 \item \textbf{(Comparison with the literature)} Theorem~\ref{eq:linear} can be used to show that the forward Euler integration of the embedding system~\eqref{eq:closedloop-embedding} with $E=E^{\mathrm{Lin}}$ and with a small enough time-step will lead to the identical over-approximation sets as~\citep[Lemma IV.3]{ME-GH-CS-JPH:21} applied to the forward Euler discretization of the linear system.  
 

 \item \textbf{(Generality of the approach)} Theorem~\ref{thm:main} proposes a general embedding-based framework for verification of the closed-loop system~\eqref{eq:closedloop} which is based on combining mixed monotone reachability of the open-loop system with interval analysis of the neural network. Using this perspective, our approach can be applied to arbitrary neural network verification algorithms as long as one can construct an inclusion function for the neural network. This is in contrast with most of the existing approaches for neural network closed-loop reachability analysis, which are heavily dependent on the neural network verification algorithm (see for instance~\citep{ME-GH-CS-JPH:21} and~\citep{CS-AM-AI-MJK:22}).
 
 \item \textbf{(Computational complexity)} From a computational perspective, the framework presented in Theorem~\ref{thm:main} consists of two main ingredients: (i) evaluating CROWN to compute the inclusion function of the neural network as in Theorem~\ref{thm:crown-rectangle}, and (ii) integrating the embedding dynamical system~\eqref{eq:closedloop-embedding}. For a neural network with $k$-layer and $N$ neuron per layer, the complexity of CROWN is $\mathcal{O}(k^2N^{3})$~\citep{HZ-etal:18}. Moreover, the functions $E^{\mathrm{G}}$ and $E^{\mathrm{H}}$ call CROWN once per integration step, while the function $E^{\mathrm{L}}$ calls CROWN $n$ times per integration step. The run time of the integration process depends on the form of the open-loop decomposition function $\OF$. 

 \end{enumerate}
\end{remark}

 
 
\section{Efficient reachability analysis via partitioning}

In this section, we develop a suitable partitioning of the state space and combine it with Theorem~\eqref{thm:main} to obtain a computationally efficient algorithm for generating reachable set over-approximations of the closed-loop system~\eqref{eq:closedloop}. 
%
%
Our partitioning strategy consists of two main components: (i) a uniform division of the state space of the embedding system~\eqref{eq:closedloop-embedding} to compute the neural network inclusion functions using CROWN, and (ii) a uniform subdivision to implement the integration on the embedding system~\eqref{eq:closedloop-embedding}. We first pick an $s\in \{\mathrm{G},\mathrm{H},\mathrm{L}\}$ and start with the initial perturbation set $\mathcal{X}_0$. In the first step, we find  the smallest hyper-rectangle containing $\mathcal{X}_0$ by computing $(\overline{x}_0)_j = \max_{x\in \mathcal{X}_0} x_j$ and  $(\underline{x}_0)_j = \min_{x\in \mathcal{X}_0} x_j$, for every $i\in \{1,\ldots,n\}$. We then divide the hyper-rectangle $[\underline{x}_0,\overline{x}_0]$ into $D_a$ partitions and obtain the set $\{[\underline{x}^1,\overline{x}^1],\ldots,[\underline{x}^{D_a},\overline{x}^{D_a}]\}$. For every $k\in \{1,\ldots, D_a\}$, we compute the trajectory of the embedding system~\eqref{eq:closedloop-embedding} with $E = E^{\mathrm{s}}$ and the initial condition $\left[\begin{smallmatrix}\underline{x}^k\\ \overline{x}^k\end{smallmatrix}\right]$ at time $\Delta t$ as in Theorem~\ref{thm:main}. For every $k\in \{1,\ldots, D_a\}$, we divide each states of the hyper-rectangle $[\underline{x}^k,\overline{x}^k]$ into $D_s$ partitions, obtaining the subpartitions  $\{[\underline{x}^{k,1},\overline{x}^{k,1}],\ldots,[\underline{x}^{k,D_s},\overline{x}^{k,D_s}]\}$. For every $k\in \{1,\ldots,D_a\}$ and $l\in \{1,\ldots,D_s\}$, we compute the trajectory of the embedding system~\eqref{eq:closedloop-embedding} with $E = E^{\mathrm{Bs}}$  defined by 
\begin{align}\label{eq:BH}
\begin{bmatrix}
   \underline{E}_i^{\mathrm{Bs}}(x,\widehat{x},w,\widehat{w})\\
\overline{E}_i^{\mathrm{Bs}}(x,\widehat{x},w,\widehat{w})
\end{bmatrix}
   = \begin{bmatrix}
\OF_i(x,\widehat{x},\underline{\OG}_{[\underline{x}^{k},\overline{x}^{k}]}(x,\widehat{x}_{[i:x]}),\overline{\OG}_{[\underline{x}^{k},\overline{x}^k]}(\widehat{x}_{[i:x]},x), w,\widehat{w})\\
    \OF_i(\widehat{x}, x,\underline{\OG}_{[\underline{x}^{k},\overline{x}^{k}]}(x_{[i:\widehat{x}]},\widehat{x}),\overline{\OG}_{[\underline{x}^{k},\overline{x}^k]}(\widehat{x},x_{[i:\widehat{x}]}), \widehat{w},w)
    \end{bmatrix},
\end{align}
and with the initial condition $\left[\begin{smallmatrix}\underline{x}^{k,l}\\ \overline{x}^{k,l}\end{smallmatrix}\right]$ at time $\Delta t$ as in Theorem~\ref{thm:main}. We then set $\mathcal{X}_1= \bigcup_{k=1}^{D_a} \bigcup_{l=1}^{D_s} [\underline{x}^{k,l}(\Delta t),\overline{x}^{k,l}(\Delta t)]$ and repeat this procedure on the initial set $\mathcal{X}_1$. We keep repeating this algorithm until we get to the final time $T$. Note that our partitioning approach is different from~\citep{ME-GH-JPH:21} and~\citep{WX-HDT-XY-TTJ:21} in that we re-partition the state-space at every time step. A summary of the above procedure is presented in Algorithm~\ref{alg:cap}. 

\begin{algorithm}[ht]
\caption{Over-approximation of reachable sets of~\eqref{eq:closedloop}}\label{alg:cap}
 \hspace*{\algorithmicindent} \textbf{Input:} $s\in \{\mathrm{G},\mathrm{H},\mathrm{L}\}$, the initial set $\mathcal{X}_0$, the final time $T$, the actuation step $\Delta t$, divisions $D_a$, subdivision $D_s$  \\
 \hspace*{\algorithmicindent} \textbf{Output:} the over-approximation of the reachable set $\overline{\mathcal{R}}(T,0,\mathcal{X}_0,\mathcal{W})$ 
\begin{algorithmic}[1]
\State $j \gets 0$
\While{$j < \left\lfloor \frac{T}{\Delta t} \right\rfloor$}
    \State $\underline{x}_i \gets \min_{x\in \mathcal{X}_j} x_i$, for every $i\in \{1,\ldots,n\}$
    \State $\overline{x}_i  \gets \max_{x\in \mathcal{X}_j} x_i$, for every $i\in \{1,\ldots,n\}$
    \State $\{[\underline{x}^1,\overline{x}^1],\ldots,[\underline{x}^{D_a},\overline{x}^{D_a}]\} \gets$ uniform\_partition$([\underline{x},\overline{x}], D_a)$ 
     \vspace{0.1cm}
    \For{$k=\{1,\ldots,D_a\}$}
        \State Compute $\underline{A}(\underline{x}^k,\overline{x}^k)$, $\overline{A}(\underline{x}^k,\overline{x}^i)$, $\underline{b}(\underline{x}^k,\overline{x}^k)$, and $\overline{b}(\underline{x}^k,\overline{x}^k)$ using CROWN and~\eqref{eq:ABcrown}. 
        \State $\{[\underline{x}^{k,1},\overline{x}^{k,1}],\ldots,[\underline{x}^{k,D_s},\overline{x}^{k,D_s}]\} \gets$ uniform\_partition$([\underline{x}^k,\overline{x}^k], D_s)$ 
        \vspace{0.1cm}
        \For{$l=\{1,\ldots,D_s\}$}
            \State Compute $\left[\begin{smallmatrix}\underline{x}^{k,l}(\Delta t)\\\overline{x}^{k,l}(\Delta t)\end{smallmatrix}\right]$ for system~\eqref{eq:closedloop-embedding} with $E=E^{\mathrm{Bs}}$ and initial condition $\left[\begin{smallmatrix}\underline{x}^{k,l}\\\overline{x}^{k,l}\end{smallmatrix}\right]$.
        \EndFor
    \EndFor
    \State $\mathcal{X}_{j+1} = \bigcup_{k=1}^{D_a} \bigcup_{l=1}^{D_s} [\underline{x}^{k,l}(\Delta t),\overline{x}^{k,l}(\Delta t)]$
\EndWhile
\State \Return $\overline{\mathcal{R}}(T,0,\mathcal{X}_0,\mathcal{W})\gets \mathcal{X}_{\left\lfloor \frac{T}{\Delta t} \right\rfloor}$
\end{algorithmic}
\end{algorithm}

\section{Numerical Simulations}

In this section, we show the efficiency of our reachability analysis using numerical experiments on a nonlinear vehicle model and a linear quadrotor model\footnote{All the code is available at \newline\url{https://github.com/gtfactslab/L4DC2023_NNControllerReachability}}. 


\subsection{Nonlinear Vehicle model}
We consider the dynamics of a vehicle adopted from~\citep{PP-FA-BdAN-AdlF:17} satisfying the following nonlinear ordinary differential equation: 
\begin{align}\label{eq:vehicle}
    \dot{p_x} &= v \cos(\phi + \beta(u_2)),\quad\dot{\phi}=\frac{v}{\ell_r}\sin(\beta(u_2)),\quad
    \dot{p_y} &= v \sin(\phi + \beta(u_2)),\quad\dot{v}= u_1 + w,
\end{align}
where $[p_x,p_y]^{\top}\in \real^2$ is the displacement of the center of mass, $\phi \in [-\pi,\pi)$ is the heading angle in the plane, $v\in \real^+$ is the speed of the center of mass. Control input $u_1$ is the applied force subject to disturbance $w$, input $u_2$ is the angle of the front wheels, and $\beta(u_2) = \mathrm{arctan}\left(\frac{\ell_f}{\ell_f+\ell_r}\tan(u_2)\right)$ is the slip slide angle. We set $x=[p_x,p_y,\phi,v]^\top$ and $u=[u_1,u_2]^\top$. We use the following decomposition function for the open-loop system:
\begin{align*}
    \OF({x},\widehat{x},{u},\widehat{u},{w},\widehat{w}) &=
    \begin{bmatrix} 
        d^{b_1b_2}\left([{v},d^{\mathrm{cos}}({\phi}+\beta({u}_2), \widehat{\phi}+\beta(\widehat{u}_2))]^\top,  [\widehat{v},d^{\mathrm{cos}}(\widehat{\phi}+\beta(\widehat{u}_2),{\phi}+\beta({u}_2))]^\top\right) \\
        d^{b_1b_2}\left([{v},d^{\mathrm{sin}}({\phi}+\beta({u}_2), \widehat{\phi}+\beta(\widehat{u}_2))]^\top, [\widehat{v},d^{\mathrm{sin}}(\widehat{\phi}+\beta(\widehat{u}_2),{\phi}+\beta({u}_2))]^\top\right) \\
        d^{b_1b_2}\left([{v},d^{\mathrm{sin}}(\beta({u}_2), \beta(\widehat{u}_2))]^\top, [\widehat{v},d^{\mathrm{sin}}(\beta(\widehat{u}_2), \beta({u}_2))]^\top\right) \\
        {u_1} + {w}
    \end{bmatrix},
\end{align*}
where $d^{b_1b_2}$, $d^{\mathrm{cos}}$, and $d^{\mathrm{sin}}$ are defined in~\citep{MEC-MB-SC:22}.
\begin{figure}[ht]
    \centering
    \includegraphics[width=\linewidth]{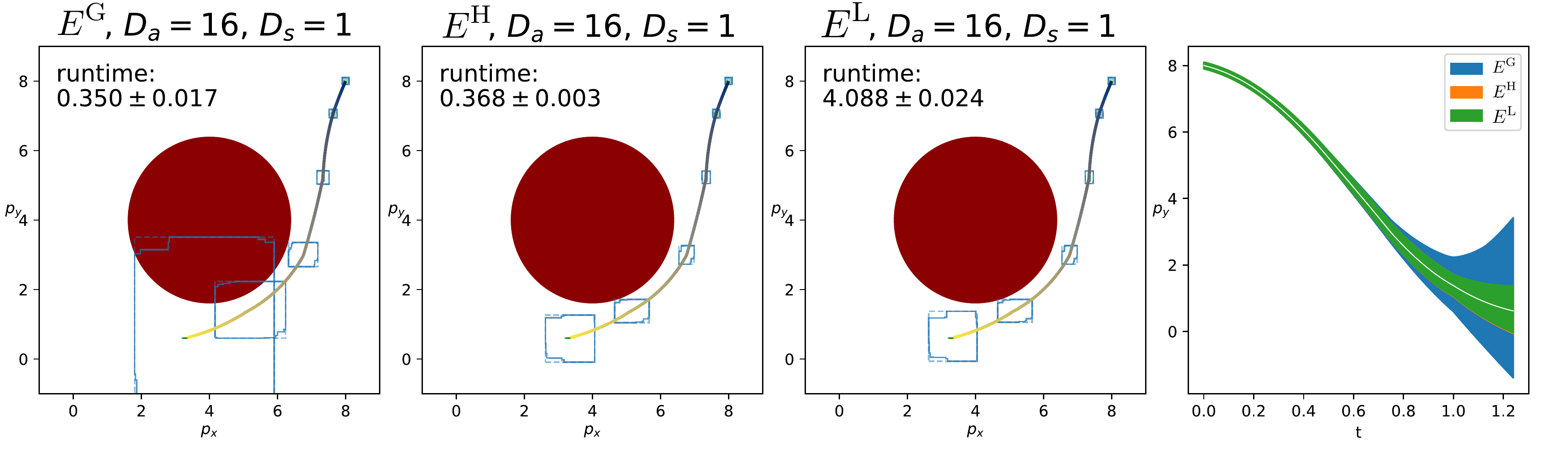}
\vspace{-0.7cm}
    \caption{Performance of the three different functions $E^{\mathrm{G}}$, $E^{\mathrm{H}}$, and $E^{\mathrm{L}}$ in Theorem~\eqref{thm:main} for over-approximation of the reachable set of the system~\eqref{eq:vehicle} with neural network $u = \ON(x)$ trained to approximate an offline model predictive controller. The $p_x-p_y$ plot of the motion of the vehicle is shown starting from an initial set  $[7.9, 8.1]^2\times[-\tfrac{2\pi}{3} - 0.01,-\tfrac{2\pi}{3} + 0.01]\times[1.99,2.01]$. The size of the over-approximations obtained using the functions $E^{\mathrm{H}}$ and $E^{\mathrm{L}}$ are close to each other, but are much smaller than the size of the over-approximations obtained using the function $E^{\mathrm{G}}$. On the other hand, finding the over-approximations using the functions $E^{\mathrm{G}}$ and $E^{\mathrm{H}}$ take about the same amount of time and are much faster than finding the over-approximations using the functions $E^{\mathrm{L}}$.  The runtimes are averaged over $10$ instances and mean and standard deviation are reported.}
    \label{fig:v_main_experiment}
    \vspace{-0.5cm}
\end{figure}
We designed an offline nonlinear model predictive controller in Python using Casadi~\citep{JAE-JG-GH-JBR-MD:19} to steer the vehicle to the origin while avoiding obstacles. We use a fixed horizon of 20 with an actuation step of 0.25 seconds, and a quadratic cost function with $Q=\mathrm{diag}(1,1,0,0)$, $Q_{hor}=\mathrm{diag}(100,100,0,1)$, and other regularizing terms tuned empirically. Additionally, we add circular obstacles with 25\% padding as hard constraints with slack variables; in Figures \ref{fig:v_main_experiment} and \ref{fig:v_partitioning_experiment}, we consider one centered at $(4,4)$ with a radius of 2.4. We simulated 65000 real trajectories (5s, 20 control actions) with initial conditions uniformly sampled from a specified region, and aggregated the data into a set of $1.3M$ training pairs $(x,u)\in\real^4\times \real^2$. A neural network $u=\ON(x)$ with $2$ hidden layers with $100$ neurons per layer and ReLU activation was trained in Pytorch to approximate the model predictive controller under a scaled Mean Squared Error loss. We use Algorithm~\ref{alg:cap} with $D_a=16$ and $D_s=1$ to provide over-approximations for the reachable sets of the vehicle model~\eqref{eq:vehicle}, comparing functions $E^{\mathrm{G}}$, $E^{\mathrm{L}}$, and $E^{\mathrm{H}}$ from Theorem~\ref{thm:main}. Algorithm~\ref{alg:cap} line 7 is computed using auto\_LiRPA~\citep{xu2020automatic}. The results are shown in Figure \ref{fig:v_main_experiment}.




\begin{figure}[ht]
    \centering
    \includegraphics[width=\linewidth]{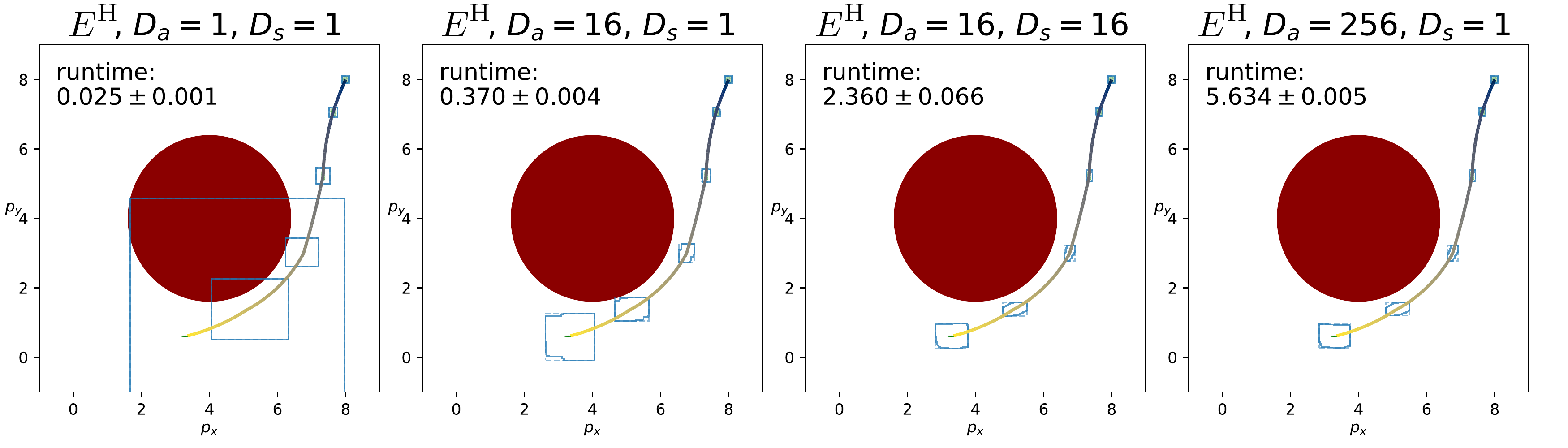}
    \vspace{-0.7cm}
    \caption{Performance of the Algorithm~\eqref{alg:cap} with $s={\mathrm{H}}$ for different partitions $D_s$ and sub-partitions $D_a$ for over-approximation of the reachable set of the system~\eqref{eq:vehicle} with neural network $u = \ON(x)$ trained to approximate an offline model predictive controller. All four figures show the $p_x-p_y$ plot of the motion of the vehicle. The runtimes are averaged over $10$ instances and the mean and standard deviation are reported.}
    \label{fig:v_partitioning_experiment}
    \vspace{-0.7cm}
\end{figure}


\subsection{Linear 6D quadrotor model}

We use the linear 6D quadrotor model adopted from~\citep{RI-JW-RA-GJP-IL:19} with the dynamics $\dot{x} = A x + B u + c $ where $A,B,c,u$ are as define in~\citep{ME-GH-CS-JPH:21}  
and $x=[p_x,p_y,p_z,v_x,x_y,v_z]^{\top}$ is the state of the system with $p_x,p_y$, and $p_y$ being the linear displacement in the $x,y$, and $z$ directions, respectively and $v_x,v_y$, and $v_z$ the velocity in the $x,y$, and $z$ directions respectively. The neural network controller $u=\ON(x)$ consists of $2$ layers with $32$ neurons in each layer and is identical to~\cite[Section H]{ME-GH-CS-JPH:21}. We use Theorem~\ref{thm:main} with the embedding function $E^{H}$ to over-approximate the reachable sets of the closed-loop system. The results are shown in Figure~\ref{fig:JonHow-comparison}. 


\begin{figure}[ht]
    \centering
    \includegraphics[width=\linewidth]{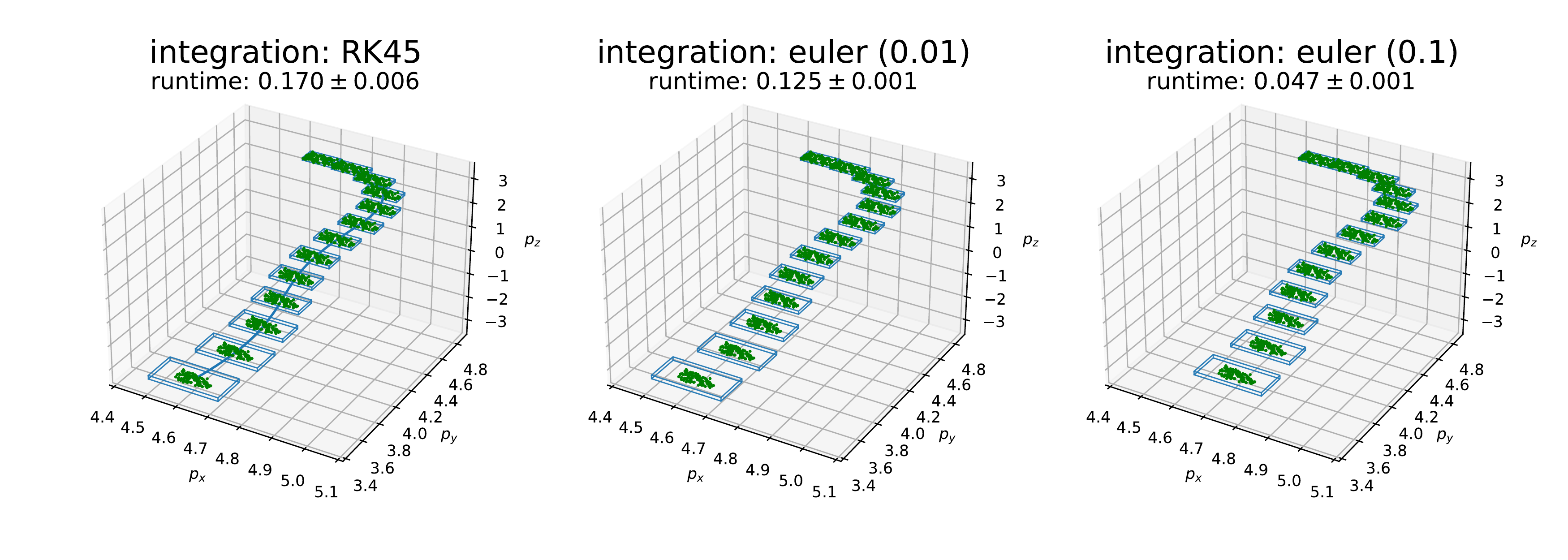}
    \vspace{-0.9cm}
    \caption{The time evolution of the reachable sets of the 6D quadrotor with the neural network controller $u=\ON(x)$ is shown in $p_x,p_y,p_z$ coordinate starting from the initial set $[4.65, 4.75]^2\times [2.95, 3.05]\times [0.94,0.96]\times [-0.01,0.01]^2$. The blue hyper-rectangles are the over-approximations of the reachable sets of the system computed using the function $E^{\mathrm{Lin}}$ in Theorem~\ref{eq:linear}. \textbf{Left plot}: The ODEs are integrated using a Runge-Kutta method. 
    \textbf{Middle plot}: The ODEs are integrated using forward Euler method with time step $0.01$.  
    \textbf{Right plot}: The ODEs are integrated using forward Euler method with time step $0.1$ to match the actuation time step. Note that the reachable set estimates of this plot are very similar to those in~\cite[Figure 10(a)]{ME-GH-CS-JPH:21}, which takes $0.113\pm 0.004$ seconds on the same computer. The runtimes are averaged over $10$ instances and mean and standard deviation are reported.}  
    
    \label{fig:JonHow-comparison}
    \vspace{-0.7cm}
\end{figure}

\section{Conclusions}

We presented a fast and scalable method, based on mixed monotone theory, to over-approximate the reachable sets of nonlinear systems coupled with neural network controllers. We introduce three methods to intertwine neural network inclusion functions from CROWN~\citep{HZ-etal:18} and open-loop decomposition functions with varying empirical results. For future research, we plan to study the role of the neural network verification algorithms in the tightness of our approximations. 

\acks{This work was supported in part by Cisco Systems and the National Science Foundation under awards \#1931980 and \#2219755.}

\bibliography{SJ,Main,FB,alias,fullalias}

\end{document}